\documentclass[10pt, twocolumn, comsoc]{IEEEtran}

\usepackage{graphicx,epsfig}
\usepackage[noadjust]{cite}
\usepackage{mcite}
\usepackage{amsfonts,helvet}
\usepackage{fancyhdr}
\usepackage{threeparttable}
\usepackage{epsf,epsfig}
\usepackage{amsthm}
\newtheorem{prop}{Proposition}
\usepackage{amsmath}
\usepackage{siunitx}
\usepackage{amssymb}
\usepackage{stfloats}
\usepackage{comment}
\usepackage{booktabs}

\usepackage[colorlinks=true, linkcolor=blue]{hyperref}

\usepackage{dsfont}
\usepackage{subfigure}
\usepackage{color}
\usepackage{enumerate}
\usepackage{gensymb}
\usepackage{cancel}
\usepackage{lipsum}
\usepackage{mathtools}
\usepackage{cuted}
\usepackage{bbm}
\usepackage[linesnumbered,ruled]{algorithm2e}
\usepackage{hyperref}
\usepackage{cleveref}

\newtheorem{remark}{Remark}

\usepackage{eucal}

\begin{document}

\title{Attention-Aided MMSE with Ridge Denoising: How to Train under Noisy Channel Samples}

\author{TaeJun Ha, Hyeji Kim, and Jeonghun Park

\thanks{This work was supported in part by the 6GARROW project, which has received funding from the Smart Networks and Services Joint Undertaking (SNS JU) under the European Union's Horizon Europe research and innovation programme (Grant Agreement No. 101192194), and in part by the Institute for Information \& Communications Technology Planning \& Evaluation (IITP) grants funded by the Korea government (MSIT): (No. RS-2024-00435652), (No. RS-2024-00395824, Development of Cloud Virtualized RAN (vRAN) System Supporting Upper-Midband), and (No. IITP-2025-RS-2024-00428780, 6G Cloud Research and Education Open Hub).

T. Ha and J. Park are with the School of Electrical and Electronic Engineering, Yonsei University, Seoul, South Korea (e-mail:{\texttt{ tjha@yonsei.ac.kr, jhpark@yonsei.ac.kr}}). 

H. Kim is with the Department of Electrical and Computer Engineering at the University of Texas at Austin, Austin, TX, USA (e-mail:{\texttt{ hyeji.kim@austin.utexas.edu}}).
}}

\maketitle \setcounter{page}{1} 
\begin{abstract} 
    Deep neural channel estimators are typically trained with clean channel state information (CSI), which is unavailable in practical orthogonal frequency-division multiplexing (OFDM) systems. In pilot-based OFDM, naive noisy-target training is structurally biased because the pilot input and noisy full-grid target share the same noise realization, driving the estimator toward identity copying. To address this for Attention-aided MMSE (A-MMSE), we propose a ridge-regularized objective that penalizes the generated filter directly. In a stylized fixed-filter model, this penalty induces scalar shrinkage and recovers the scalar MMSE gain at an explicit penalty value. We further construct surrogate training targets by estimating the channel covariance via eigenvalue clipping of the noisy empirical second-moment matrix, without requiring clean CSI labels. On COST 2100 channels, the proposed Ridge-A-MMSE consistently outperforms the Noise2Noise (N2N) baselines considered in this paper, and the combination of ridge regularization and covariance shrinkage approaches the same network trained with clean CSI labels at high signal-to-noise ratios (SNRs).
\end{abstract}


\section{Introduction} 

In orthogonal frequency-division multiplexing (OFDM) communication systems, accurate channel estimation is crucial as it underpins coherent demodulation. 
The classical linear minimum mean squared error (LMMSE) estimator requires explicit knowledge of the channel's second-order statistics, which are difficult to obtain in practice. 
Deep neural network (DNN)-based methods circumvent this limitation by learning the channel directly from data. 
These estimators have evolved from end-to-end networks~\cite{2018Ye} and convolutional super-resolution networks~\cite{2019Soltani} to attention-based estimators~\cite{2023Luan, ha2026learningmmsefiltersofdm, guo2024parallel} and to deep generative priors~\cite{arvinte2023score, fesl2024diffusion}. Across these architectural advances, however, the training methodology has remained the same: clean channel state information (CSI) labels, or clean channel realizations for fitting a generative prior,
are assumed available during training.
This assumption is difficult to justify in practical deployments, where only noisy observations of the channel are available through received pilot signals. 
Thus, the relevant problem is clean-label-free training: only received pilot observations and noisy full-grid channel samples,  collected offline in dedicated sounding slots, are available during training, while the desired objective remains clean-channel mean squared error (MSE) minimization.

Networks trained on synthetically designed channels, for which exact labels are available by construction, have been shown to generalize robustly to unseen channels~\cite{2023luanICC, 2026luanTCCN}.
This paper targets the complementary measurement-driven regime, in which site-specific channel statistics must be learned from noisy over-the-air observations.
In this regime, a representative label-free baseline is the Noise2Noise (N2N) principle from image restoration~\cite{2018N2N}: two independent noisy observations of the same scene serve as input and target, yielding unbiased training when the noise realizations are independent.

In pilot-based OFDM, however, only a single noisy snapshot is available per slot~\cite{3GPP-TS-38.211}, so reusing it as both the training input and the training target makes the input noise and the target noise identical at every pilot position, violating the independence condition required by N2N. 
A recent label-free approach beyond N2N~\cite{yu2024bayes} restores input--target noise independence by injecting synthetic noise into a complete noisy channel estimate. In our pilot-based setting, by contrast, the target coincides with the input at every pilot coordinate, and estimation requires interpolation rather than mere denoising.
The noisy-target loss is therefore structurally misaligned with clean-channel MSE minimization: it encourages copying the noisy pilot input rather than denoising it, motivating a training methodology that requires neither clean nor independently observed target labels.

Among DNN-based OFDM channel estimators, Attention-aided MMSE (A-MMSE) is a model-based framework that employs an Attention Transformer to learn a linear MMSE filter, rather than directly mapping pilots to channel estimates~\cite{ha2026learningmmsefiltersofdm}.
Attention captures the global, non-stationary grid correlation that the local, shift-invariant kernels of convolutional super-resolution networks~\cite{2019Soltani} recover only implicitly.
Crucially for clean-label-free training, the network output is the linear filter itself rather than a channel estimate. The filter is thus an explicit, interpretable object that a training objective can penalize directly---a handle that end-to-end convolutional estimators do not expose. 
At inference, channel estimation reduces to a single matrix–vector multiplication with no nonlinear activations.
Despite these advantages, A-MMSE still requires clean CSI labels for training.


In this paper, building on A-MMSE’s linear inference structure, we propose Ridge-A-MMSE for clean-label-free OFDM channel estimation under noisy channel samples. 
Section~\ref{sec: model, form} formalizes the noisy-target bias described above.
The proposed method combines two complementary mechanisms.
First, on the filter side, Ridge-A-MMSE augments the A-MMSE training objective with a Frobenius-norm penalty on the generated linear filter. 
We show that, in the stylized fixed-filter model of Section~\ref{sec: method}, this penalty recovers the scalar MMSE gain at an explicit penalty value.
Second, on the target side, we introduce eigenvalue-clipped covariance shrinkage to construct LMMSE-style plug-in surrogate targets from the noisy full-grid corpus, requiring no clean CSI.
Numerical results on COST 2100 channels~\cite{COST2100} demonstrate that the proposed Ridge-A-MMSE consistently outperforms the considered N2N baselines~\cite{2018N2N} across both semi-urban and high-speed railway scenarios over the full $0$–$35$ dB signal-to-noise ratio (SNR) range. 
Against the N2N-trained counterpart with the identical backbone, the normalized MSE (NMSE) reduction is $57.3\%$ at $20$ dB and reaches $86.5\%$ at $35$ dB.
Since the N2N baselines reuse the same observation as input and target at 35 dB (Section~\ref{sec: results}), the $86.5\%$ margin is measured directly against naive same-slot training, the biased regime analyzed in Section~\ref{sec: model, form}.
At high SNR, the proposed method approaches the network trained with clean CSI labels.

\textbf{Notations}: 
$(\cdot)^{\sf{T}}$ and $(\cdot)^{\sf{H}}$ denote transpose and Hermitian transpose.
$\mathbb{E}[\cdot]$ denotes the expectation operator.
Operators $\operatorname{vec}(\cdot)$ and $\operatorname{diag}(\cdot)$ denote matrix vectorization and the diagonal matrix formed from a vector, respectively.
$\lVert \cdot \rVert_F$ and $\lVert \cdot \rVert_2$ represent the Frobenius norm and $\ell_2$-norm.
$\mathbf{I}_N$ is the $N \times N$ identity matrix, and $\mathbf{e}_i$ denotes the $i^{\mathrm{th}}$ standard basis vector. 
$\mathbf{A} \succeq 0$ indicates that $\mathbf{A}$ is positive semidefinite, and $x \sim \CMcal{CN}(0, \Sigma)$ denotes a complex Gaussian vector with zero mean and covariance matrix $\Sigma$.

\section{Problem Setup and Learning Objective}\label{sec: model, form}
\subsection{System Model and A-MMSE Inference}\label{sec: model & inference}
We consider an OFDM slot with $N$ subcarriers and $M$ symbols. 
The full channel is $\mathbf{h}=\operatorname{vec}(\mathbf{H})\in\mathbb C^{NM}$, where $\mathbf{H} \in \mathbb{C}^(N \times M)$ collects the channel coefficients. 
Let $\CMcal P=\{p_1,\ldots,p_L\}$ denote the fixed set of pilot coordinates, let $\mathbf{h}_p\in\mathbb{C}^L$ be the pilot-restricted channel vector, and let $\mathbf{x}_p \in \mathbb{C}^L$ collect the transmitted pilot symbols. 
With the known pilot matrix
$\mathbf{D}_p=\operatorname{diag}(\mathbf{x}_p)$ and additive pilot noise $\mathbf{z}_p\sim\CMcal{CN}(\mathbf{0},\sigma_n^2\mathbf{I}_L)$,
the received pilot vector is
\begin{align}\label{eq: pilotobs}
    \mathbf{y}_p = \mathbf{D}_p \mathbf{h}_p + \mathbf{z}_p .
\end{align}
For unit-modulus pilots, we equivalently use the de-rotated pilot observation $\mathbf{r}_p=\mathbf{D}_p^{\sf{H}}\mathbf{y}_p=\mathbf{h}_p+\bar{\mathbf{z}}_p$,
where $\bar{\mathbf{z}}_p \triangleq \mathbf{D}_p^{\sf H}\mathbf{z}_p \sim \mathcal{CN}(\mathbf{0}, \sigma_n^2\mathbf{I}_L)$.
A-MMSE generates a linear filter $\mathbf{W}_\theta(\mathbf{y}_p)\in\mathbb{C}^{NM\times L}$ and forms
\begin{align}\label{eq:ammse_est}
    \hat{\mathbf{h}} = \mathbf{W}_{\theta}(\mathbf{y}_p)\, \mathbf{y}_p . 
\end{align}

The filter in~\eqref{eq:ammse_est} is generated by the A-MMSE attention network of~\cite{ha2026learningmmsefiltersofdm}, depicted in the attention detail of Fig.~\ref{fig:pipeline}. The real and imaginary parts of $\mathbf{y}_p$ are embedded into $2L$ tokens.
Learned projections yield $\mathbf{Q},\mathbf{K} \in \mathbb{R}^{2L\times d_k}$ and $\mathbf{V}\in \mathbb{R}^{2L\times d_v}$. 
The attention map $\mathbf{A}=\operatorname{softmax}\!\big(\mathbf{Q}\mathbf{K}^{\mathsf{T}}\!/\sqrt{d_k}\big)\in\mathbb{R}^{2L\times 2L}$ scores pairwise correlation among these pilot components. It is thus a learned, input-adaptive counterpart of the pilot-domain second-order statistics that LMMSE filtering requires. 
The attended features $\mathbf{A}\mathbf{V}$, after the output projection $\mathbf{W}_O$, are decoded into the filter $\mathbf{W}_\theta(\mathbf{y}_p)$ penalized in Section~\ref{sec: method}. Following~\cite{ha2026learningmmsefiltersofdm}, this block is applied in two stages to capture frequency and temporal correlations.

\subsection{Clean-Label-Free Setting and the Noisy-Target Bias}\label{sec: setting and target}
In the clean-label-free setting considered in this paper, the training stage has access only to
\begin{align}\label{eq:dataset}
     \CMcal{D} = \{(\mathbf{y}_p^{(i)},\tilde{\mathbf{h}}^{(i)})\}_{i=1}^{K},\; \tilde{\mathbf{h}}^{(i)} = \mathbf{h}^{(i)} + \mathbf{n}^{(i)},
 \end{align}
where $\mathbf{n}^{(i)}\sim\CMcal{CN} (\mathbf{0},\sigma_n^2\mathbf{I}_{NM})$. 
The noise variance $\sigma_n^2$ is assumed known throughout; unlike the channel statistics, it can be estimated reliably at the receiver.
Such full-grid samples can be collected in an offline phase in which every resource element carries a known pilot (e.g., dedicated sounding slots): de-rotating the received full-grid observation then yields $\tilde{\mathbf{h}}^{(i)} = \mathbf{h}^{(i)} + \mathbf{n}^{(i)}$, while at deployment only the sparse pilot set $\CMcal{P}$ is transmitted.
In the corpus of~\eqref{eq:dataset}, the pilot input is the sounding observation restricted to $\CMcal{P}$, so the label noise at the pilot coordinates is the pilot-noise realization itself; equivalently,
\begin{align}\label{eq: pilot form}
    \tilde{\mathbf{h}}_p^{(i)} = \mathbf{h}_p^{(i)}+\bar{\mathbf{z}}_p^{(i)} = \mathbf{r}_p^{(i)}.
\end{align}

Although only noisy training samples are available, the desired estimator is the one that minimizes the clean-channel MSE. 
In population form, the ideal A-MMSE learning objective is the clean risk
\begin{align}\label{eq:dataset_objective}
    \theta^\star = \arg\min_{\theta} \CMcal{R}_{\mathrm{clean}}(\theta),
    \,\,
    \CMcal{R}_{\mathrm{clean}}(\theta) = \mathbb{E} \left[\left\| \mathbf{h} -  \mathbf{W}_\theta(\mathbf{y}_p)\mathbf{y}_p  \right\|_2^2 \right].
\end{align}
Since $\mathbf{h}^{(i)}$ is unavailable, the problem is to construct a trainable surrogate objective from $\CMcal{D}$ that remains aligned with~\eqref{eq:dataset_objective}.


A natural surrogate replaces the unavailable clean label $\mathbf{h}$ with the noisy observation $\tilde{\mathbf{h}}$:
\begin{align}\label{eq:naive}
    \CMcal{L}_{\text{naive}}(\theta) = \mathbb{E}\left[\lVert \tilde{\mathbf{h}} - \mathbf{W}_\theta(\mathbf{y}_p)\mathbf{y}_p \rVert_2^2\right].
\end{align}
However, $\CMcal{L}_{\mathrm{naive}}(\theta)$ is not simply a noisy estimate of $\CMcal{R}_{\rm clean}(\theta)$. 
By~\eqref{eq: pilot form}, at pilot coordinate $p_i$,
\begin{align}\label{eq:sharednoise}
    r_i=h_{p_i}+\bar z_{p_i},
    \qquad
    \tilde h_{p_i}=h_{p_i}+\bar z_{p_i}=r_i .
\end{align}
Thus, the noisy target equals the input itself at pilot coordinates.

For the scalar same-coordinate estimator $\hat{h}_{p_i}=\alpha r_i$, the naive noisy-target objective gives
\begin{align} \label{eq:naivealpha}
    \alpha_{\mathrm{naive}} =\arg\min_\alpha \mathbb{E} \left[ \left|\alpha r_i-\tilde{h}_{p_i} \right|^2  \right]  = \arg\min_\alpha \mathbb{E} \left[\left| \alpha r_i-r_i \right|^2 \right]  = 1 .
\end{align}
In contrast, the clean-MSE-optimal scalar gain is
\begin{align} \label{eq:mmsealpha}
    \alpha_{\mathrm{MMSE}} = \arg\min_\alpha \mathbb{E} \left[ \left| \alpha r_i-h_{p_i}\right|^2 \right] = \frac{\sigma_h^2}{\sigma_h^2+\sigma_n^2} <1 .
\end{align}
The naive objective thus favors identity copying $\alpha_{\mathrm{naive}}=1$, whereas clean MSE minimization requires shrinkage. This motivates the filter-side ridge penalty in Section~\ref{sec: method} and the target-side covariance-shrinkage surrogate in Section~\ref{sec: processing}.


\section{Ridge-Regularized A-MMSE}\label{sec: method}


Ridge-A-MMSE counters the identity-copying bias of Section~\ref{sec: setting and target} with a filter-side correction: it penalizes the generated A-MMSE filter itself while preserving the inference rule in~\eqref{eq:ammse_est}.

For the raw noisy-target case, this gives
\begin{align}\label{eq:ridgeobj}
    \CMcal{L}_{\mathrm{ridge}}(\theta) = \underbrace{\mathbb{E}\left[\lVert \tilde{\mathbf{h}} - \mathbf{W}_\theta(\mathbf{y}_p)\mathbf{y}_p \rVert_2^2\right]}_{\CMcal{L}_{\mathrm{naive}}(\theta)}
    + \lambda\,\mathbb{E}\left[\lVert \mathbf{W}_\theta(\mathbf{y}_p) \rVert_{\sf F}^2\right],
\end{align}
where $\lambda > 0$ is the ridge parameter. 
This is a Tikhonov-style regularization~\cite{hoerl1970ridge, tikhonov1977solutions} applied directly to the generated filter, not to internal DNN weights; it supplies the shrinkage pressure that the naive objective lacks.

To gain analytical insight into this filter-side correction, we analyze a stylized setting that deliberately abstracts away the sample dependence of the learned filter: a fixed linear filter $\mathbf{W} \in \mathbb{C}^{NM \times L}$, independent of $\mathbf{y}_p$, a single pilot coordinate $p_i \in \CMcal{P}$, and isotropic de-rotated pilot covariance.
The analysis is stated in the de-rotated domain; both the channel estimate and the ridge penalty are invariant to this choice, as shown in Section~\ref{sec: processing}.
Let $\mathbf{w}_k^{\sf{H}}$ denote row $k$ of $\mathbf{W}$, and let $\mathbf{r} \triangleq \mathbf{r}_p = \mathbf{h}_p + \bar{\mathbf{z}}_p$ be the de-rotated pilot vector. From \eqref{eq:sharednoise}, the noisy target at pilot coordinate $p_i$ satisfies $\tilde{h}_{p_i} = r_i$, so the ridge objective at this row reduces to
\begin{align}\label{eq:rowwise}
    \min_{\mathbf{w}_k}\;\mathbb{E}\left[|\mathbf{w}_k^{\sf{H}} \mathbf{r} - r_i|^2\right] + \lambda \lVert \mathbf{w}_k \rVert_2^2.
\end{align}
The following proposition shows that, in this stylized setting, the ridge penalty induces scalar shrinkage and recovers the scalar MMSE gain at an explicit penalty value $\lambda^\star$.

\begin{prop}[Scalar shrinkage in a fixed-filter model]\label{prop:ridge}
Assume $\mathbb{E}[h_k] = \mathbb{E}[n_k] = 0$, $\mathbb{E}[h_k n_k^*]=0$, $\mathbb{E}[|h_k|^2] = \sigma_h^2$, $\mathbb{E}[|n_k|^2] = \sigma_n^2$, and isotropic de-rotated pilot covariance $\mathbf{R}_{\mathbf{rr}}  = \sigma_y^2\mathbf{I}_L$, with $\sigma_y^2 \triangleq \sigma_h^2 + \sigma_n^2.$
For the pilot-output row $k = p_i$ and any $\lambda > 0$, the minimization
in~\eqref{eq:rowwise} admits the unique minimizer
$\mathbf{w}_{p_i,\lambda} = \alpha_\lambda \mathbf{e}_i$, where
$\alpha_\lambda = \sigma_y^2/(\sigma_y^2+\lambda)$.
In this same-coordinate scalar case, the corresponding test MSE is
\begin{align}\label{eq:parabola}
    J(\alpha) = \mathbb{E}\left[|\alpha r_i - h_{p_i}|^2\right] = (1-\alpha)^2\sigma_h^2+\alpha^2\sigma_n^2.
\end{align}
Hence $J(\alpha_{\lambda})<\sigma_n^2$ for every $\lambda>0$ when $\sigma_h^2\le \sigma_n^2$; when $\sigma_h^2 > \sigma_n^2$, the same improvement holds for
\begin{align}\label{eq:lrange}
    0<\lambda<\bar\lambda,\;
    \bar\lambda \triangleq \frac{2\sigma_n^2(\sigma_h^2+\sigma_n^2)}{\sigma_h^2-\sigma_n^2}.
\end{align}
The scalar MMSE point $\alpha_{\mathrm{MMSE}} = \sigma_h^2 / (\sigma_h^2 + \sigma_n^2)$ is recovered at
\begin{align}\label{eq:lstar}
    \lambda^{\star}=\frac{\sigma_n^2(\sigma_h^2+\sigma_n^2)}{\sigma_h^2}.
\end{align}
\end{prop}

\begin{proof}[Proof sketch]
For the stylized fixed-filter problem with $\mathbf{W}$ independent of $\mathbf{y}_p$, the quadratic objective is separable across the rows of $\mathbf{W}$; thus the analysis reduces to the pilot-output row associated with $p_i$.
With ${\bf{R}}_{{\bf{rr}}} = \sigma^2_y {\bf{I}}_L$, the first-order condition 
\(
(\mathbf{R}_{\mathbf{rr}} + \lambda \mathbf{I}_L)\mathbf{w}_k = \mathbb{E}[\mathbf{r} \mathbf{r}_i^*] = \sigma^2_y \mathbf{e}_i
\)
yields the unique minimizer $\mathbf{w}_{k,\lambda} = \alpha_\lambda \mathbf{e}_i$.
The closed form in \eqref{eq:parabola} follows by expanding the expectation and using $\mathbb{E}[h_{p_i}\bar{z}_{p_i}^{*}] = 0$ to eliminate the cross-term. Evaluating $J(\alpha_\lambda) - \sigma_n^2$ establishes the improvement range, and solving $\alpha_\lambda = \alpha_{\mathrm{MMSE}}$ for $\lambda$ yields $\lambda^\star$.

\end{proof}

\begin{figure}[t] 
    \centering
    \includegraphics[width=0.9\linewidth]{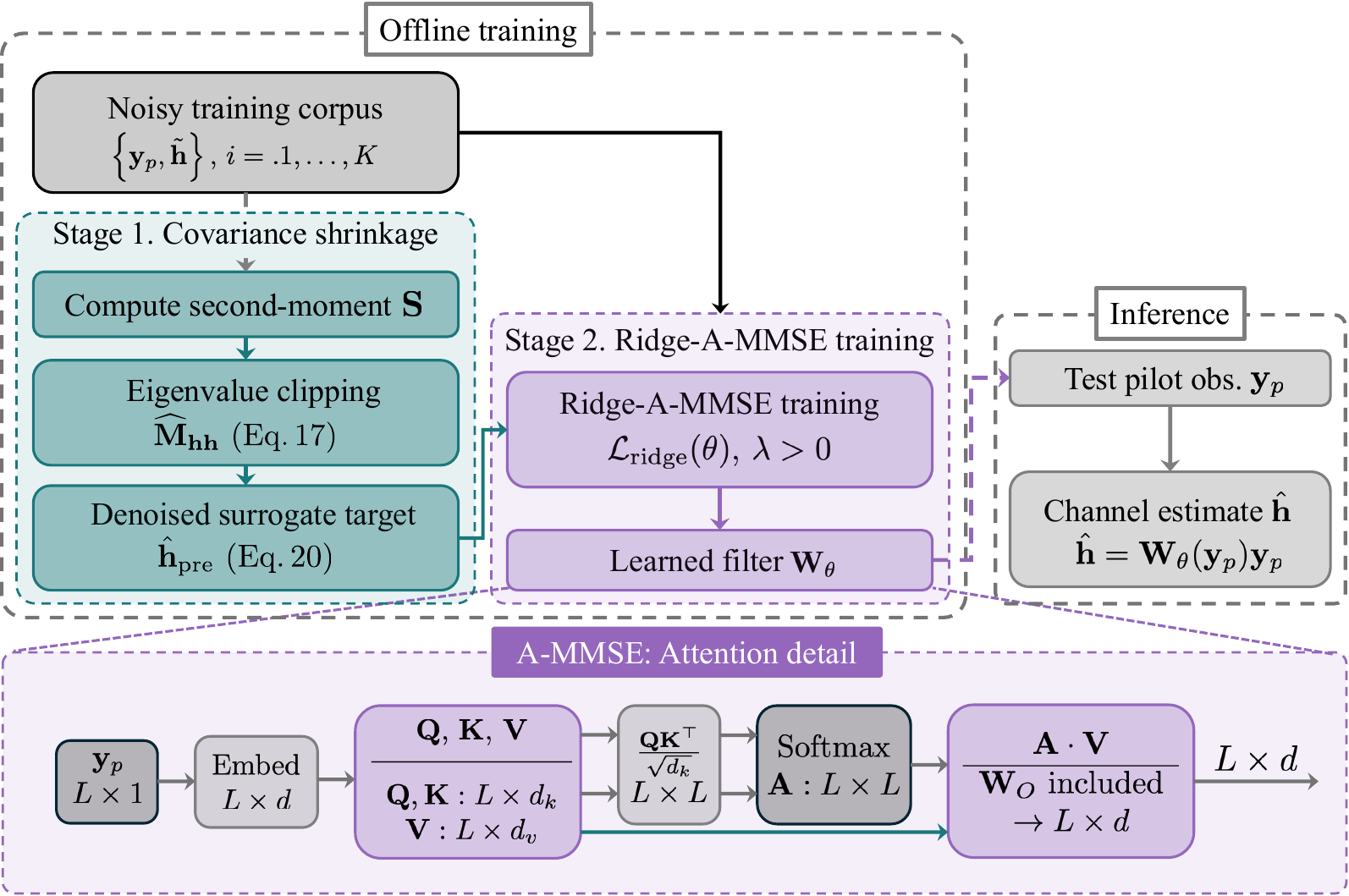} 
    \caption{Overview of the proposed Ridge-A-MMSE training and inference pipeline.}
    \label{fig:pipeline}
\end{figure}

Without the penalty ($\lambda = 0$), the minimizer $\mathbf{w}_{p_i} = \mathbf{e}_i$ copies the noisy pilot and attains $J(1) = \sigma_n^2$ from \eqref{eq:parabola}, reproducing the identity bias of Section~\ref{sec: model, form}. For any $\lambda \in (0, \bar{\lambda})$, the shrinkage $\alpha_\lambda < 1$ strictly improves on this baseline, and the scalar MMSE gain $\alpha_{\mathrm{MMSE}}$ is recovered exactly at $\lambda^\star$.
Since $\lambda^\star / \bar{\lambda} = (\sigma_h^2 - \sigma_n^2)/(2\sigma_h^2) < 1/2$, the MMSE point lies well inside the improvement range, so, within this stylized model, tuning $\lambda$ near $\lambda^\star$ helps rather than hurts on this row.

\begin{remark}\normalfont \label{rem: prop1}
    Proposition~\ref{prop:ridge} is a mechanism-level explanation of why the Frobenius penalty counteracts the identity-copying bias of Section~\ref{sec: model, form}. 
    It provides no optimality guarantee for the learned architecture. In Ridge-A-MMSE, $\mathbf{W}_{\theta}(\mathbf{y}_p)$ is generated by the attention network and varies with the pilot observation, and the de-rotated pilot covariance is in general correlated. 
    The penalty is applied per sample, so the shrinkage pressure acts on every generated filter, although its closed-form characterization holds only in the stylized model. The effect of the penalty under correlated covariances and learned filters is assessed empirically in Section~\ref{sec: results}. 
\end{remark}

During training, $\theta$ is optimized by minimizing~\eqref{eq:ridgeobj}, with $\mathbf{y}_p$ as the input and $\tilde{\mathbf{h}}$ as the label; at inference, the channel estimate is formed exactly as in \eqref{eq:ammse_est}, confining the ridge modification to the training stage.
A label-free rule for choosing $\lambda$, derived from this mechanism, is given in Section~\ref{sec: results}.

\section{Surrogate Targets via Covariance Shrinkage} \label{sec: processing}


The ridge penalty of Section~\ref{sec: method} discourages identity copying but does not remove the noise in $\tilde{\mathbf{h}}$ itself. 
We therefore introduce a complementary target-side correction that constructs an LMMSE-style plug-in surrogate target from the noisy full-grid corpus.
As illustrated in Fig.~\ref{fig:pipeline}, Stage 1 builds the surrogate targets via~\eqref{eq:samplecov} -- \eqref{eq:surrogatetarget}, and Stage 2 trains A-MMSE under the ridge-regularized objective~\eqref{eq: training objective}.




The channel is assumed zero-mean in this section; in practice, the sample mean is subtracted before forming $\mathbf{S}$. From the noisy corpus, we form the sample second-moment matrix
\begin{align}\label{eq:samplecov}
    \mathbf{S} = \frac{1}{K}\sum_{i=1}^{K} \tilde{\mathbf{h}}^{(i)}\tilde{\mathbf{h}}^{(i)\sf{H}}.
\end{align}
Assuming that $\mathbf{h}$ and $\mathbf{n}$ are uncorrelated, 
\begin{align} \label{eq:secondmoment}
    \mathbb{E}[\mathbf{S}] = \mathbf{M}_{\mathbf{h}\mathbf{h}}+\sigma_n^2 \mathbf{I}_{NM},\; \mathbf{M}_{\mathbf{h}\mathbf{h}} \triangleq  \mathbb{E}\left[\mathbf{h}\mathbf{h}^{\sf{H}}\right].
\end{align}
Now let $\mathbf{S} = \mathbf{U}\operatorname{diag}(\mu_1,\ldots,\mu_{NM}) \mathbf{U}^{\sf{H}}$ be the eigendecomposition of $\mathbf{S}$. 
Since~\eqref{eq:secondmoment} is an expectation-level identity, subtracting the nominal noise floor $\sigma_n^2$ from the empirical eigenvalues is a plug-in shrinkage step rather than an exact sample-wise correction.
Empirically, smaller eigenvalues of $\mathbf{S}$ are more likely noise-dominated and may become negative after the subtraction; clipping them to zero enforces positive semidefiniteness:
\begin{align}\label{eq:clip}
    \widehat{\mathbf{M}}_{\mathbf{hh}}=\mathbf{U}\operatorname{diag}\left((\mu_1-\sigma_n^2)_+,\ldots,(\mu_{NM}-\sigma_n^2)_+\right) \mathbf{U}^{\sf{H}},
\end{align}
where $(x)_+ = \max(x, 0)$. The construction is consistent with classical eigenvalue-shrinkage covariance estimators~\cite{ledoit2004well, donoho2018optimal}.

We now use $\widehat{\mathbf{M}}_{\mathbf{hh}}$ to build the surrogate target that replaces $\tilde{\mathbf{h}}$ in the training objective. With the de-rotated pilot vector $\mathbf{r}_p = \mathbf{h}_p +  \bar{\mathbf{z}}_p$ from Section~\ref{sec: model, form}, the MMSE linear map is obtained as
\begin{align}\label{eq:prefilter}
    \mathbf{W}_{\mathrm{pre}}^\star = \arg\min_{\mathbf{W}}\, \mathbb{E}\left[\|\mathbf{h} - \mathbf{W}\mathbf{r}_p\|_2^2\right],
\end{align}
which leads to 
\begin{align}\label{eq:preprocessedfilter}
    \mathbf{W}_{\mathrm{pre}}^\star = \mathbf{M}_{\mathbf{h}\mathbf{h}_p}\left(\mathbf{M}_{\mathbf{h}_p\mathbf{h}_p} + \sigma_n^2 \mathbf{I}_L\right)^{-1},
\end{align}
where $\mathbf{M}_{\mathbf{h}\mathbf{h}_p}$ and $\mathbf{M}_{\mathbf{h}_p\mathbf{h}_p}$ denote the corresponding cross- and pilot-domain blocks of $\mathbf{M}_{\mathbf{hh}}$.
Since only the noisy samples in $\CMcal{D}$ from~\eqref{eq:dataset} are available, the true $\mathbf{M}_{\mathbf{hh}}$ is unknown; plugging in $\widehat{\mathbf{M}}_{\mathbf{hh}}$ from~\eqref{eq:clip} produces a surrogate target of the LMMSE form evaluated at the estimated covariance:

\begin{align}\label{eq:surrogatetarget}
    \hat{\mathbf{h}}_{\mathrm{pre}} = \widehat{\mathbf{M}}_{\mathbf{h}\mathbf{h}_p}\left(\widehat{\mathbf{M}}_{\mathbf{h}_p\mathbf{h}_p} + \sigma_n^2 \mathbf{I}_L\right)^{-1} \mathbf{r}_p.
\end{align}
We write $\widehat{\mathbf{W}}_{\mathrm{pre}}$ for the plug-in filter in \eqref{eq:surrogatetarget}; $\hat{\mathbf{h}}_{\mathrm{pre}} = \widehat{\mathbf{W}}_{\mathrm{pre}}\,\mathbf{r}_p$ is thus a deterministic linear function of the pilot input, yet unlike the raw target of \eqref{eq:sharednoise} it steers the generated filter toward a denoising map rather than identity copying.
The ridge term in \eqref{eq: training objective} excludes an exact copy of this map as the minimizer, while $\mathbf{W}_{\theta}(\mathbf{y}_p)$ remains input-adaptive.
No sample-wise denoising guarantee is claimed; for a fixed $\mathbf{W}$ independent of the noise, $\mathbb{E}\big\|\mathbf{W}\mathbf{r}_p - \mathbf{h}\big\|_2^2 = \mathbb{E}\big\|\mathbf{W}\mathbf{h}_p - \mathbf{h}\big\|_2^2 + \sigma_n^2 \|\mathbf{W}\|_{\mathrm{F}}^2 < \mathbb{E}\big\|\tilde{\mathbf{h}} - \mathbf{h}\big\|_2^2 = NM\sigma_n^2$
once $\sigma_n^2\big(NM - \|\mathbf{W}\|_{\mathrm{F}}^2\big)$ exceeds the pilot interpolation error---a condition more easily met for noisier corpora.
The final Ridge-A-MMSE training objective with covariance-shrinkage surrogate targets is therefore
\begin{align}\label{eq: training objective}
    \theta_{\mathrm{RA}} = \arg\min_{\theta}\!\frac{1}{K}\sum_{i=1}^{K}\!\left[\!\left\|\hat{\mathbf{h}}_{\mathrm{pre}}^{(i)}-\mathbf{W}_{\theta}(\mathbf{y}_p^{(i)})\,\mathbf{y}_p^{(i)}\right\|_2^2 +\lambda\left\|\mathbf{W}_{\theta}(\mathbf{y}_p^{(i)})\right\|_F^2 \right].
\end{align}


The surrogate~\eqref{eq:surrogatetarget} is stated in the de-rotated domain of~\eqref{eq:prefilter}--\eqref{eq:preprocessedfilter}, while the objective~\eqref{eq: training objective} keeps the raw input $\mathbf{y}_p$, matching~\eqref{eq:ammse_est} and~\eqref{eq:ridgeobj}; for unit-modulus pilots $\mathbf{D}_p$ is unitary, so $\mathbf{W}_\theta(\mathbf{y}_p)\mathbf{y}_p = \left[ \mathbf{W}_\theta(\mathbf{y}_p)\mathbf{D}_p \right]\mathbf{r}_p$ and $\lVert \mathbf{W}_\theta(\mathbf{y}_p)\mathbf{D}_p \rVert_F = \lVert \mathbf{W}_\theta(\mathbf{y}_p) \rVert_F$, making both the estimate and the penalty invariant to the domain choice.

\begin{remark}\normalfont
    The clipping in~\eqref{eq:clip} is the Frobenius-norm projection of $\mathbf{S} - \sigma_n^2\mathbf{I}_{NM}$ onto the positive semidefinite cone~\cite{higham1988computing}. Three properties hold at every corpus size $K$. First,~\eqref{eq:secondmoment} gives $\mathbb{E}[\mathbf{S}] - \sigma_n^2\mathbf{I}_{NM} = \mathbf{M}_{\mathbf{hh}}$, so the construction is exact at the population level. 
    Second, the projection is non-expansive, $\|\widehat{\mathbf{M}}_{\mathbf{hh}} - \mathbf{M}_{\mathbf{hh}}\|_F \le \|\mathbf{S} - \mathbb{E}[\mathbf{S}]\|_F$, and clipping does not amplify the finite-sample error of $\mathbf{S}$. 
    Third, $\widehat{\mathbf{M}}_{\mathbf{hh}} \succeq 0$ gives $\widehat{\mathbf{M}}_{\mathbf{h}_p\mathbf{h}_p} + \sigma_n^2\mathbf{I}_L \succeq \sigma_n^2\mathbf{I}_L$, so the inverse in~\eqref{eq:surrogatetarget} exists and has spectral norm at most $\sigma_n^{-2}$.
\end{remark}

\begin{remark}[Complexity]\normalfont
    The additional computation of Ridge-A-MMSE relative to A-MMSE~\cite{ha2026learningmmsefiltersofdm} is confined to offline training. In Stage 1, forming $\mathbf{S}$ in~\eqref{eq:samplecov} costs $\CMcal{O}(K(NM)^2)$, and its eigendecomposition costs $\CMcal{O}((NM)^3)$. Both are performed once per training corpus.
    Constructing the plug-in filter in~\eqref{eq:surrogatetarget} costs $\CMcal{O}((NM)^2L + L^3)$, and generating the $K$ surrogate targets costs $\CMcal{O}(KNML)$. For larger grids, the eigendecomposition can be replaced by a truncated decomposition without altering the training pipeline. In Stage 2, the ridge penalty in~\eqref{eq: training objective} adds $\CMcal{O}(NML)$ per sample, negligible compared with the attention forward and backward passes. At inference, the estimator is identical to~\eqref{eq:ammse_est}.
    The filter is generated by the unchanged A-MMSE backbone and applied through a single $\CMcal{O}(NML)$ matrix--vector product. Ridge-A-MMSE thus incurs no additional inference-time complexity, latency, or memory relative to A-MMSE.
\end{remark}

\section{Simulation Results}\label{sec: results}

\begin{figure*}[t]
    \centering
    \subfigure[NMSE performance on SU scenario.]{
        \includegraphics[width=0.47\textwidth]{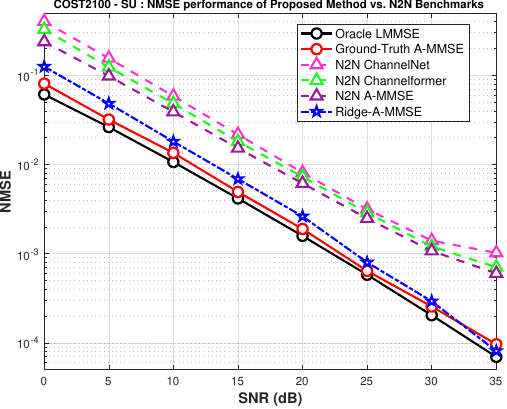}
        \label{fig: NMSE_SU}}
    \hfill
    \subfigure[NMSE performance on HSR scenario.]{
        \includegraphics[width=0.47\textwidth]{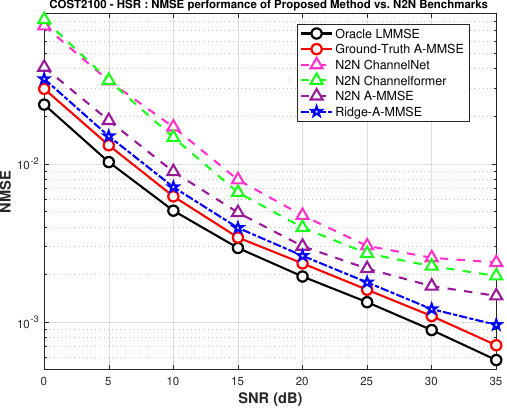}
        \label{fig: NMSE_HSR}}
    \caption{NMSE comparison of the proposed Ridge-A-MMSE and the N2N baselines on COST 2100 channels. N2N baselines require $35$ dB target observations, whereas Ridge-A-MMSE trains only on the noisy corpus at each operating SNR.}
    \label{fig: Results}
\end{figure*}

\textit{\textbf{Setup and baselines:}} The proposed Ridge-A-MMSE is evaluated on two COST 2100 single-user channel scenarios~\cite{COST2100}.
The semi-urban (SU) scenario is configured at $3.5$ GHz carrier frequency, $30$ kHz subcarrier spacing, $1000$ ns delay spread, $40$ km/h user speed, and Rician factor of $3$ dB (maximum Doppler shift $\approx 129.6$ Hz).
The high-speed railway (HSR) scenario operates at $5$ GHz carrier frequency with $60$ kHz subcarrier spacing, $100$ ns delay spread, $350$ km/h user speed, and Rician factor of $13$ dB (maximum Doppler shift $\approx 1620.4$ Hz).
Both scenarios use an OFDM slot of $N = 72$ subcarriers and $M = 14$ symbols, with $L= 72$ pilots following the pilot pattern of~\cite{ha2026learningmmsefiltersofdm}, whose temporal train-validation-test split of $36000$/$4000$/$4000$ slots is also adopted.
All models are trained with Adam~\cite{kingma2015adam} on a corpus of $K = 36000$ noisy slots per scenario and operating SNR, on a single NVIDIA RTX 4090 GPU.
Fig.~\ref{fig: Results} reports the NMSE, defined as $\mathbb{E}\|\hat{\mathbf{h}}-\mathbf{h}\|_2^2 \,/\, \mathbb{E}\|\mathbf{h}\|_2^2$ over the test set, for $0$--$35$ dB SNR in both scenarios.
The proposed Ridge-A-MMSE uses the A-MMSE backbone with both the ridge penalty and the eigenvalue-clipped surrogate targets.

Two categories of reference methods are considered.
Oracle LMMSE serves as a classical linear benchmark that assumes perfect knowledge of the true channel statistics, while Ground-Truth A-MMSE represents the clean-label reference, i.e., the same backbone trained with clean CSI labels; both are idealized and unavailable in practice.
Three N2N baselines --- N2N ChannelNet~\cite{2019Soltani}, N2N Channelformer~\cite{2023Luan}, and N2N A-MMSE --- are trained following the N2N principle~\cite{2018N2N}. 
The two end-to-end deep-learning (DL) baselines are faithful implementations of the original architectures, and N2N Channelformer follows the original offline version. Only the input--output dimensions are adapted to the resource grid and pilot pattern of this paper. 

Each N2N training pair combines the operating-SNR pilot input with a $35$ dB full-grid observation of the same channel realization: the $35$ dB level reflects that target observations would be gathered offline under favorable conditions, while the independent second observation of each realization is an idealization granted to these baselines.
Below $35$ dB, the input and target noise are thus statistically independent, as N2N requires, making the loss unbiased in expectation, although the residual target noise still inflates its finite-sample variance; at $35$ dB, the same observation necessarily serves as both, and N2N reduces to the same-slot regime of Section~\ref{sec: setting and target}.
Ridge-A-MMSE is instead trained at each operating SNR, using the noisy corpus observed at that SNR for both the pilot inputs and the surrogate-target construction; no separately collected observation is required at any operating point.

\textit{\textbf{Ridge parameter:}} The ridge parameter is set per scenario and operating SNR; a label-free rule follows from Proposition~\ref{prop:ridge}, aggregated over the filter rows.
The penalty trades a first-order gain on the pilot rows against a second-order cost elsewhere.
On the $L$ pilot-output rows, the raw target of~\eqref{eq:ridgeobj} coincides with the input by~\eqref{eq:sharednoise}, so a small penalty reduces the clean risk at a rate proportional to $\sigma_n^2$ per row -- the derivative of~\eqref{eq:parabola} at $\alpha = 1$.
On the remaining rows, the target noise is independent of the input, so the objective is already aligned with the clean risk.
The penalty there inflicts only a second-order loss, weighted by the squared row norms.
Balancing the two at high SNR and unit channel power gives $\hat{\lambda} =L\sigma_n^2 / \lVert  \widehat{\mathbf{W}}_{\mathrm{pre}}\rVert_F^2$, a plug-in rule of the classical Hoerl–Kennard–Baldwin form~\cite{hoerl1975ridge} computed from Stage-1 statistics alone.
The values in Fig.~\ref{fig: Results} match $\hat{\lambda}$ within a small constant factor, and the NMSE varies by less than $1$ dB over more than a decade of $\lambda$ around them.


\textit{\textbf{Results:}} In the SU scenario in Fig.~\ref{fig: NMSE_SU}, the proposed method yields improvements over all N2N baselines across the entire SNR range.
Compared with N2N A-MMSE, which shares the same network architecture, Ridge-A-MMSE achieves $53.5\%$ lower NMSE at $10$ dB, $57.3\%$ at $20$ dB, and $86.5\%$ at $35$ dB.
Against the end-to-end DL baselines, Ridge-A-MMSE improves over N2N ChannelNet by $69.0\%$ at $10$ dB and $67.7\%$ at $20$ dB, while the corresponding gains over N2N Channelformer are $62.5\%$ and $64.1\%$.
At $35$ dB, the proposed method reaches an NMSE of $8.094 \times 10^{-5}$, slightly better than the Ground-Truth A-MMSE value of $9.642 \times 10^{-5}$, which we attribute to the implicit regularization provided by the covariance-shrinkage surrogate target and the ridge penalty.
In the HSR scenario in Fig.~\ref{fig: NMSE_HSR}, the proposed method outperforms all N2N baselines, though the relative margins are narrower than in the SU case.
Ridge-A-MMSE achieves $20.1\%$ and $13.0\%$ lower NMSE than N2N A-MMSE at $10$ dB and $20$ dB, respectively, with the improvement increasing to $34.4\%$ at $35$ dB.
Compared with N2N Channelformer, relative gains of $33.8\%$ at $20$ dB and $50.9\%$ at $35$ dB are obtained.

These gains hold under a stricter training regime that requires no independently drawn target observation. 
Because N2N A-MMSE shares the identical backbone, the gap to it isolates the effect of the training methodology alone; the margins over the end-to-end DL baselines additionally include the contribution of the A-MMSE architecture. 
In the SU channel the proposed method approaches the clean-label reference at high SNR, and in the HSR channel it attains the lowest NMSE among the methods trained exclusively from noisy observations.


\textit{\textbf{Ablation and convergence:}} Under the settings of Fig.~\ref{fig: Results} at the full corpus size $K = 36000$, over both scenarios and all operating SNRs, the ridge penalty alone yields $23.8$–$47.9$$\%$ lower NMSE than naive same-slot noisy-target training, and the surrogate targets alone up to $74.1\%$ lower. The full method yields $24.8$–$47.8\%$ lower NMSE than the better single mechanism at every point from $0$ to $30$ dB.
At $35$ dB it remains below the better single mechanism in SU and ties the ridge-only arm in HSR, where the surrogate advantage has narrowed.
Removing the eigenvalue clipping increases the measured NMSE by up to $204.8\%$ at $0$ dB, while the clipping is inactive at $35$ dB.
A further ablation removes the attention encoder from the A-MMSE backbone of~\cite{ha2026learningmmsefiltersofdm}, with all other components and the entire training procedure unchanged. The full method attains $58.5$--$93.0\%$ (SU) and $25.2$--$64.1\%$ (HSR) lower NMSE than the resulting FC-only (fully connected only) variant across the entire SNR range.
Against a parameter-matched variant, in which each attention encoder is replaced by residual fully connected blocks at the total parameter count of the full backbone under identical training, the margin remains $40.7$--$90.3\%$ (SU) and $15.3$--$39.8\%$ (HSR).

The training and validation losses of the proposed method converge stably at $35$ dB in both scenarios, with a final train--validation gap below $4\%$. The converged validation loss, measured against the clean channel, is $86.3\%$ (SU) and $35.3\%$ (HSR) lower than that of naive same-slot noisy-target training.
Re-evaluated against the surrogate target -- the reference computable in the field -- the converged validation loss agrees with its clean-referenced value within about $12\%$ (SU) and $1.2\%$ (HSR), so convergence can be monitored without clean labels. 

\textit{\textbf{Computational cost:}} The recurring computational cost of the proposed method matches that of the A-MMSE backbone. One training epoch over the $36000$-slot corpus takes approximately $30$ s for the A-MMSE variants, $26$ s for N2N ChannelNet, and $18$ s for N2N Channelformer. Stage 1 of the proposed method adds a one-time preprocessing cost on the order of seconds per corpus. 
At inference, every A-MMSE variant applies the trained filter through the single matrix–vector product in~\eqref{eq:ammse_est}. The measured runtime is approximately $0.37$ ms per slot on the RTX 4090, versus approximately $26$ ms and $21$ ms for N2N ChannelNet and N2N Channelformer.
The complete experimental results, including the per-arm ablation and the full set of ridge parameters, are available online~\cite{A-MMSE-GitHub}.

\section{Conclusion}\label{sec: conclusion}


We proposed Ridge-A-MMSE for clean-label-free DNN-based OFDM channel estimation, where training must rely on pilot observations and noisy full-grid channel samples alone.
The method addresses the resulting noisy-target bias from two complementary directions.
On the filter side, ridge regularization penalizes the generated A-MMSE filter and, in a stylized fixed-filter pilot-coordinate model, induces scalar shrinkage that recovers the scalar MMSE gain at an explicit penalty value.
On the target side, eigenvalue-clipped covariance shrinkage constructs LMMSE-style plug-in surrogate targets from the noisy corpus without requiring clean CSI labels.
Simulations on COST 2100 channels showed that the combination of ridge regularization and covariance shrinkage consistently outperforms the considered N2N baselines across all SNRs and approaches Ground-Truth A-MMSE at high SNR. 
Future work includes a theoretical characterization of the label-free ridge rule beyond the stylized setting of Proposition~\ref{prop:ridge}, and an extension of the framework to multi-antenna and multi-user MIMO settings.
Another direction is a generalization study across channel families, combining designed-data pre-training~\cite{2023luanICC, 2026luanTCCN} with the proposed clean-label-free on-site refinement.

\bibliographystyle{IEEEtran}
\bibliography{ridge_AMMSE_ref}

\end{document}